\documentclass[a4paper,10pt]{article}
\usepackage[utf8]{inputenc}

\usepackage{url}
\usepackage{amsmath,amssymb,pgfplots,amsthm,authblk}
\usepackage[hidelinks]{hyperref}

\newcommand{\dg}{d_g}
\newcommand{\Lg}{{L}_g}
\newcommand{\vol}{\operatorname{vol}_g}
\newcommand{\tvol}{\tau_V}

\newcommand{\Ric}{\operatorname{Ric}}
\newcommand{\bN}{\mathbb{N}}

\newcommand{\bR}{\mathbb{R}}

\theoremstyle{plain}
\newtheorem{theorem}{Theorem}
\newtheorem{proposition}{Proposition}
\newtheorem{lemma}{Lemma}

\theoremstyle{definition}
\newtheorem{example}{Example}
\newtheorem{remark}{Remark}
\newtheorem{definition}{Definition}

\title{Mean Curvature, Singularities and Time Functions in Cosmology}

\author[1]{Gregory J. Galloway}
\author[2]{Leonardo Garc\'ia-Heveling}

\affil[1]{\small Department of Mathematics, University of Miami,
Coral Gables, FL 33124, USA, Email:~\texttt{galloway@math.miami.edu}}
\affil[2]{Fachbereich Mathematik, Universit\"at Hamburg,
Bundesstra{\ss}e~55,~20146 Hamburg, Germany,
Email:~\texttt{leonardo.garcia@uni-hamburg.de}}
\date{}

\begin{document}

\maketitle

\begin{abstract}
 In this contribution, we study spacetimes of cosmological interest, without making any symmetry assumptions. We prove a rigid Hawking singularity theorem for positive cosmological constant, which sharpens known results. In particular, it implies that any spacetime with $\Ric \geq ng$ in timelike directions and containing a compact Cauchy hypersurface with mean curvature $H \geq n$ is timelike incomplete. We also study the properties of cosmological time and volume functions, addressing questions such as: When do they satisfy the regularity condition? When are the level sets Cauchy hypersurfaces? What can one say about the mean curvature of the level sets? This naturally leads to consideration of Hawking type singularity theorems for Cauchy surfaces  satisfying mean curvature inequalities in a certain weak sense.
 \medskip

\end{abstract}

\section{Introduction}

A fundamental assumption in cosmology, the study of the universe at large scales, is that our position as observers in the Universe is not special. This translates into the hypothesis that the universe is approximately homogeneous and isotropic. According to general relativity, it can then be modelled by a FLRW spacetime, that is, a Lorentzian manifold of the form $(b,\infty) \times N$ with metric tensor
\begin{equation*}
 g := - dt^2 + a^2(t) h,
\end{equation*}
where $(N,h)$ is a Riemannian manifold of constant curvature. The warping function $a(t)$, also called scale factor, can be determined via the Einstein equations, when some assumption on the matter content of the universe is made. Its derivative $\dot a(t)$ corresponds to the expansion rate of the universe (which is known to be positive at present). An important feature of these models is the following: If the matter content obeys certain energy conditions (e.g.\ the strong energy condition, which implies $\ddot a(t) \le 0$ for all $t$), and the expansion is large enough at some time (e.g.\ $\dot a(t_0) \geq \epsilon > 0$), then $b > -\infty$. This is interpreted as the universe originating in a big bang, since the $t$-coordinate naturally plays the role of time, and we are concluding that the history of the universe only extends backwards until a finite $t$-value. In particular, if (by a translation, if necessary) we set $b=0$, then $t$ measures the maximum proper time experienced by an observer from the big bang until his or her present.

Of course, we know that the universe is not exactly homogeneous or isotropic. Thus the question arises if the features outlined above, which were  supported by observational evidence, are mathematically robust under perturbations. A large portion of this question was answered in the affirmative by Hawking's singularity theorem \cite{Haw66}, which states that an arbitrary spacetime satisfying the strong energy condition ($\Ric \geq 0$ in timelike directions) and containing a Cauchy hypersurface $S$ with mean curvature (aka expansion) $H \geq \epsilon > 0$, is timelike geodesically incomplete towards the past. Moreover, in this globally hyperbolic setting, the length of all timelike geodesics going from $S$ towards the past is uniformly bounded. Hawking \cite{Haw68} and Geroch \cite{Ger70} also showed that spacetimes satisfying appropriate causality conditions can be equipped with time functions, although in a highly non-unique way.

Some more nuanced aspects of these questions, however, were only resolved later. For one, modern observations indicate that the universe is expanding at an accelerated rate. This suggests a positive value for the cosmological constant, which is in violation of the strong energy condition, and hence would render Hawking's theorem inapplicable. This issue was addressed by Borde \cite[Theorem~4]{Borde94} (and later independently in \cite[Prop.~3.3]{AnGa02}), who proved that the strong energy condition can then be replaced by the lower Ricci bound $\Ric \geq n g$ in timelike directions, as long as the bound on the mean curvature of the Cauchy hypersurface is strong enough to compensate. This less stringent Ricci bound is compatible with a positive cosmological constant. (For a quite different variation on Hawking's singularity theorem, which allows for a  positive cosmological constant by only requiring the null energy condition, see \cite{GaLi18}.)

The issue of non-uniqueness of time functions was addressed by Andersson, Howard and the first author \cite{AGH98}, in their study of  what is now commonly referred to as the {\it cosmological time function}, which on appropriate ``big bang'' spacetimes represents the maximum proper time to the initial singularity, just as the $t$ coordinate does in the special case of FLRW spacetimes. (Much earlier, Wald and Yip \cite{WaYi81} introduced the cosmological time function, or rather its time dual, which they referred to as the `maximum lifetime function’.) A similar, alternative approach to defining a canonical time function, the {\it cosmological volume function}, using the volume instead of length element, was recently introduced by the second author \cite{GH23}.

In this paper, we prove the following two theorems, which in the body of the paper are restated as a series of several smaller results, along with some additional related results; see Sections \ref{sec:Hawkingthm} and \ref{sec:timefcts} for details.

\begin{theorem} \label{thm:intro1}
Let $(M,g)$ be an $(n+1)$-dimensional spacetime, and let $S$ be a smooth spacelike compact Cauchy hypersurface with mean curvature $H \geq n$.
Suppose that the lower curvature bound $\Ric \geq n g$ holds in timelike directions. Then  $(M,g)$ is past timelike geodesically incomplete.  More specifically, one has the following.
 \begin{enumerate}
\item If $H > n$ at some point of $S$, then all timelike geodesics are past-incomplete.\smallskip

\item If $H = n$, then either some geodesic normal to $S$ is past-incomplete, or else $J^-(S)$ is isometric to $((-\infty,0] \times S, -dt^2 \oplus e^{2t} h)$, in which case all nonnormal timelike geodesics are past-incomplete.\smallskip
 \end{enumerate}

\vspace{-.05in}
\noindent
Furthermore:

 \vspace{-.05in}
 \begin{enumerate}
 \item[3.] The cosmological volume function $\tvol$ of $(M,g)$ is a continuous, finite valued time function, with $\tvol \to 0$ along every past-inextendible causal curve.
 \end{enumerate}
 \end{theorem}

 \begin{remark}
 By our conventions, the mean curvature $H$ is the trace of the second fundamental form $K$, defined as $K(X,Y) = g(\nabla_X u, Y)$, for $X,Y \in T_pS$, where $u$ is the future directed unit normal to $S$. For simplicity, we assume the spacetime metric $g$ to be smooth.
 \end{remark}

The Ricci curvature condition of the above theorem will hold in a spacetime that satisfies the Einstein equation with cosmological constant
 $\Lambda = \frac{n(n-1)}{2}$ provided the energy-momentum tensor
 satisfies a natural nonnegativity condition.
Under a similarly natural nonnegativity condition, the mean curvature of a spacelike hypersurface $S$ will satisfy $H \ge n$ (or $H \le -n$) provided $S$ has nonpositive scalar curvature.

Theorem \ref{thm:intro1} is a rigid version of Hawking's singularity theorem for positive cosmological constant, rigid meaning that we allow for the case $H = n$. It builds upon previous rigidity results. In particular, point~1 and the splitting part of point~2 were proven in \cite{AnGa02}. The splitting also follows from a more general result of Graf \cite{Gra20}, applicable to certain non-compact hypersurfaces $S$.  See Section 2 for further discussion. Point 3 tells us that, moreover, the assumptions of Hawking's theorem are sufficient for having a canonical notion of time through the cosmological volume function. For the cosmological time function $\tau$, it is necessary to assume separately that it is regular (i.e.\ $\tau < \infty$ and $\tau \to 0$ along every past-inextendible causal curve). On the plus side, this regularity assumption then gives $\tau$ the following  nice properties (beyond those considered in \cite{AGH98}).

\begin{theorem} \label{thm:intro2}
 Let $(M,g)$ be an $(n+1)$-dimensional spacetime equipped with a regular cosmological time function $\tau$, with level sets $S_T := \{\tau = T \}$. Then
 \begin{enumerate}
  \item If $\Ric \geq -n \kappa g$ in timelike directions, then $S_T$ has mean curvature $H_T \leq \beta_{\kappa,T}$ in the support sense. (See \eqref{beta} for the definition of $\beta_{\kappa,T}$.)\smallskip

  \item If $(M,g)$ is future timelike geodesically complete, and $(M,g)$ has compact Cauchy hypersurfaces or the future causal boundary of $(M,g)$ is spacelike, then $S_T$ is a Cauchy hypersurface, for every $T > 0$.
 \end{enumerate}
\end{theorem}

Andersson et al.~\cite{ABBZ12} had previously used the cosmological time function as a tool to construct a CMC time function on spacetimes with constant sectional curvature. That is, a time function whose level sets have constant mean curvature (as is the case for the $t$-level sets in an FLRW spacetime). The problem of finding CMC time functions has also been studied by Gerhardt~\cite{Ger83,Ger21} under certain asymptotic assumptions.

The paper is structured as follows. In Section~\ref{sec:Hawkingthm} we prove point 1 and (a generalization of) point 2 of Theorem~\ref{thm:intro1}. In Section~\ref{sec:timefcts} we turn to time functions, proving point 3 of Theorem~\ref{thm:intro1} and Theorem~\ref{thm:intro2}, and providing additional context and examples. In particular, in Section~\ref{sec:Hsupport} we prove a version of Hawking's singularity theorem for $C^0$ Cauchy hypersurfaces obeying a mean curvature bound in the support sense.



\section{A rigid Hawking singularity theorem for positive cosmological constant} \label{sec:Hawkingthm}

In this section, we first prove points 1 and 2 of Theorem~\ref{thm:intro1}, and then we extend point 2 to non-compact Cauchy hypersurfaces $S$ that admit an $S$-ray.

\begin{theorem}\label{thm:newHawking} Let $(M,g)$ be an $(n+1)$-dimensional spacetime, and let $S$ be a compact Cauchy hypersurface with mean curvature $H \geq n$. Suppose that the lower curvature bound
$\Ric \geq n g$ holds in timelike directions. Then  $(M,g)$ is past timelike geodesically incomplete.  More specifically, one has the following.

 \begin{enumerate}
\item If $H > n$ at some point of $S$, then all timelike geodesics are past-incomplete.\smallskip

\item If $H = n$, then either some geodesic normal to $S$ is past-incomplete, or else $J^-(S)$ is isometric to $((-\infty,0] \times S, -dt^2 \oplus e^{2t} h)$, in which case all nonnormal timelike geodesics are past-incomplete. \smallskip
 \end{enumerate}
 \end{theorem}

\begin{remark} \label{rem:rescale}
 Theorem \ref{thm:newHawking} also holds if $\Ric \geq - n \kappa g$ and $H \geq n \sqrt{ \vert \kappa \vert}$ for some constant $\kappa<0$. But then, the metric $\tilde g := -\kappa g$ satisfies
 \begin{align*}
  \widetilde{\Ric} = \Ric \geq n \tilde{g}, && \tilde{H} = \vert \kappa \vert^{-1/2} H \geq n.
 \end{align*}
 Since geodesic completeness is homothety invariant, no generality is really gained.
\end{remark}

After a perturbation, point~1 of Theorem~\ref{thm:newHawking} follows from \cite[Prop.~3.3]{AnGa02}, while \cite[Prop.~3.4]{AnGa02} further shows that in the case $H \geq n$, the only exception to incompleteness of \textit{all} geodesics occurs when $J^-(S)$ is isometric to a warped product with exponential warping function (cf.~point~2). We complete the picture with the observation that these exponentially warped products are still always incomplete (although they do contain {\it some} past-complete geodesics as well). This observation was already made, in a slightly different context and without explicit proof, in \cite[Thm.~6.1]{GaVe15}. Note that a concrete example of exponentially warped product, which provides good intuition for the general case, is a coordinate patch in de Sitter spacetime \cite[Fig.~16(ii)]{HaEl73}.

\begin{proof}
Firstly, if $H > n$ at some point of $S$, then via a modified mean curvature flow, one can obtain a Cauchy hypersurface $\tilde{S}$ with mean curvature $\tilde{H} > n$ everywhere (see proof of \cite[Prop.\ 3.4]{AnGa02} for details). Then, by Hawking's singularity theorem (the version for negative Ricci, see \cite[Prop.\ 3.3]{AnGa02}), all timelike geodesics in $(M,g)$ are past incomplete.

To prove the rest of the theorem, note that by \cite[Prop.\ 3.4]{AnGa02} (or Theorem~\ref{thm:newHawking2} below), either there is at least one past-incomplete timelike geodesic normal to $S$, or $J^-(S)$ is isometric to $(-\infty,0] \times S$ with metric tensor $g = -dt^2 + e^{2t} h$, where $h$ is the induced metric on $S$. It thus suffices to show that $(-\infty,0) \times S$ with metric tensor $g = -dt^2 + e^{2t} h$ is such that all timelike geodesics not orthogonal $S$ are past incomplete.

Choose $V \in TM$ any past-directed timelike vector with basepoint in $S$, not normal to $S$. Then the projection of $V$ onto $TS$ is non-vanishing. Let $\gamma$ be the unique geodesic with $\dot\gamma(0) = V$. The projection of $\gamma$ onto $S$ is a
$1$-dimensional immersed submanifold of $S$ (possibly with self-intersections), which we can parametrize by a coordinate $x$ such that the metric $h$ restricts to $dx^2$, and such that $x(0)=0$ and $\dot x > 0$ for our geodesic $\gamma$. To get $\dot x > 0$, and hence an immersion, we have used the fact that if $\dot x(s) = 0$ for some parameter value $s$, then by uniqueness of solutions to the geodesic equation, $\gamma$ would be a $t$-line, in contradiction to our assumption on $\dot \gamma (0)$. Thus $\gamma$ is contained in a $2$-dimensional immersed submanifold $N \cong (-\infty,0) \times I$, on which the Lorentzian metric $g$ restricts to $-dt^2+e^{2t}dx$.  Here $I \subset \bR$ is an interval determined by the projection  of $\gamma$.

Because $\gamma$ is a geodesic of the ambient spacetime $M$, it must also be a geodesic of the submanifold $N$ equipped with the restricted metric. Since $\partial_x$ is a Killing vector field for $N$, we have a conserved quantity
\begin{equation} \label{eq:C}
 C = g(\partial_x,\dot\gamma) = e^{2t} \dot x.
\end{equation}
Note that $C > 0$ because of how we have chosen $\dot\gamma(0)$. Moreover, we may parametrize $\gamma$ by $g$-arclength, so that
\begin{equation} \label{eq:arc}
 -1 = g(\dot\gamma,\dot\gamma) = -\dot t^2 + e^{2t} \dot x^2.
\end{equation}
Isolating $\dot t$ from \eqref{eq:arc} and substituting \eqref{eq:C} into it, we obtain
\begin{equation} \label{eq:tdot}
 \dot t = -\sqrt{C^2 e^{-2t} + 1} < - C e^{-t},
\end{equation}
where we have chosen the overall minus sign so that $\gamma$ is past-directed. But already the solution $y(s)$ to $\dot y = - C e^{-y}$ blows up in finite time, so the solution $t(s)$ of \eqref{eq:tdot} must blow-up at least as fast. Indeed,
\begin{equation}
 y(s) = \ln \big( - C s + \tilde{C} \big).
\end{equation}
We conclude that since $t(s)$ blows up in finite affine parameter time, the geodesic $\gamma$ is incomplete.
\end{proof}

As we now discuss, point 2 of Theorem \ref{thm:newHawking} remains valid under the weaker assumption that $S$ admits a past $S$-ray (rather than $S$ being compact). By a past $S$-ray we mean a past inextendible causal geodesic $\gamma$ starting on $S$ such that each segment of
$\gamma$ maximizes the Lorentzian distance to $S$.   A future $S$-ray is defined similarly.  If, as in the theorem, $S$ is smooth and spacelike, any $S$-ray meets $S$ orthogonally, and hence is timelike.  By a standard construction, if $S$ is compact it admits a past (as well as future) $S$-ray; see e.g.\ Section \ref{sec:cauchyness}.  The conclusion of point 2 now follows from the next result, which, for convenience, is stated with the time orientation reversed.

\begin{theorem}\label{thm:newHawking2}
Let $(M,g)$ be an $(n+1)$-dimensional globally hyperbolic spacetime, and let $S$ be a smooth spacelike Cauchy hypersurface with mean curvature $H \le -n$.  Assume also that $S$ admits a future $S$-ray. Suppose that the lower curvature bound $\Ric \geq n g$ holds in timelike directions. If the normal geodesics to $S$ are future complete, $J^+(S)$ is isometric to $([0, \infty) \times S, -dt^2 \oplus e^{-2t} h)$.
\end{theorem}

\noindent Essentially the same argument as in the proof of Theorem \ref{thm:newHawking} shows that all nonnormal timelike geodesics to $S$ are future-incomplete.

Theorem \ref{thm:newHawking2} is an extension of \cite[Theorem C]{Gal89}
to the case of a positive cosmological constant.
The theorem  is, in fact, an immediate consequence of \cite[Theorem 5.12]{Gra20} (the $\kappa = -1$ case), which, in terms of completeness,  only requires that the future $S$-ray be  complete. The proof of the latter is based on properties of the Lorentzian Busemann function, the regularity theory of which is somewhat involved; see e.g.\ \cite{Esch88,Gal89b}. Here we present a fairly short proof of Theorem~\ref{thm:newHawking2}, which makes use of the more elementary, causal theoretic methods developed in \cite{GaVe15}; cf.\ \cite[Theorem 5.13]{GaVe17}.

\begin{lemma}\label{lem:split}
Let the assumptions be as in Theorem \ref{thm:newHawking2}. Suppose for each $x \in S$ the future directed timelike unit speed normal geodesic $\gamma_x$ is an $S$-ray. Then $J^+(S)$ is isometric $([0,\infty) \times S, -dt^2 \oplus e^{-2t} h)$, where $h$ is the induced metric on $S$.
\end{lemma}

\noindent This lemma was proved in \cite{GaVe17} (though in a time dual setting); for the convenience of the reader we include the proof here.

\begin{proof}[Proof of the lemma]
Consider the normal exponential map,
\begin{equation}
\psi: [0, \infty) \to J^+(S) \,, \quad \psi(t, x) = \exp_x t N  \,,
\end{equation}
where $N$ is the future directed unit normal field along $S$.
Since each normal geodesic, $t \to \gamma_x(t) = \psi(t, x)$, $x \in S$ is an $S$-ray,  no two of them can intersect. Moreover, there can be no focal point to $S$ along any $\gamma_x$.
It follows that $\psi$ is a diffeomorphism, and, up to isometry,
\begin{equation}
J^+(S) = [0,\infty) \times S\, , \quad g = -dt^2 \oplus h_t  \,,
\end{equation}
where $h_t = h_{ij}(t,x) dx^idx^j$ is the induced metric on $S_t = \{t\} \times S$.

Let $H = H(t,x)$ be the mean curvature of $S_t$.  By the Raychaudhuri equation we have,
\begin{equation}\label{ray}
\frac{\partial H}{\partial t} = -\Ric(N,N) - \frac{H^2}{n} - |\sigma|^2 \,,
\end{equation}
where now  $N = \frac{\partial}{\partial t}$, and for each $t$, $\sigma$, the {\it shear}, is the trace free part of the second form $K$ of $S_t$,
\begin{equation}
\sigma_{ij} = K_{ij} - \frac{H}{n} h_{ij} \, .
\end{equation}

Using the Ricci curvature condition in \eqref{ray}, and the mean curvature assumption, we obtain
\begin{equation}
\frac{\partial \mathcal{H}}{\partial t} \le 1 - \mathcal{H}^2 \,, \quad\mathcal{H}(0) \le -1.
\end{equation}
where $\mathcal{H} = \frac{H}{n}$.
Elementary comparison with solutions to the ODE
$h' = 1 - h^2$ then shows that, in order for $\mathcal{H}$ to be defined for all $(t,x) \in [0,\infty) \times S$ (i.e., without diverging to $-\infty$ in finite time), one must have $\mathcal{H} \equiv -1$, and hence $H \equiv  -n$.  Using this in \eqref{ray}, together with $\Ric(N,N) \ge -n$, implies that the shear vanishes, $\sigma = 0$, and hence, $K_{ij} = -h_{ij}$. Then, using the formula $K_{ij} = \frac12 \frac{\partial h_{ij}}{\partial t}$, we obtain $h_{ij}(t,x) = e^{-2t}h_{ij}(0,x)$, and the result follows.
\end{proof}

\begin{proof}[Proof of Theorem \ref{thm:newHawking2}]  Let $S_{\infty}^-$ be the {\it ray horosphere} associated to the assumed $S$-ray $t \to \gamma(t)$, $t \in [0,\infty)$,
$$
S_{\infty}^- := \partial\left(\cup_k I^-(S_k^-) \right)  \,,
$$
where $S_k^-$ is the past Lorentzian sphere of radius $k$ centered at $\gamma(k)$, see
\cite[Def.~3.20]{GaVe15}.  $S_{\infty}^-$ is  a type of limit  of Lorentzian spheres with centers along $\gamma$; it is an example of what is called an achronal limit in \cite{GaVe15}.  In the present context, it satisfies a number of properties; see \cite[Lemmas 3.21, 3.22]{GaVe15}:
\begin{enumerate}
\item[(1)] $S_{\infty}^- \subset J^-(S)$, and $\gamma(0) \in S_{\infty}^-$.\smallskip

\item[(2)] $S_{\infty}^-$ is an acausal, edgeless $C^0$ hypersurface.\smallskip

\item[(3)] $S_{\infty}^-$ admits a timelike future $S_{\infty}^-$-ray from each of its points.

\end{enumerate}

Arguing in a manner very similar to the proof of Proposition \ref{prop:Hoftau} in Section~\ref{sec:Hsupport} of the present paper, but in a time dual way, one finds that $S^-_\infty$ has support mean curvature $\ge \lim_{t \to \infty} -n\coth(t) = -n$
(see Section \ref{sec:Hsupport} for the definition of support mean curvature).  Let $S^-$ be the connected component of $S^-_\infty$ which contains $\gamma(0)$. Then $S \cap S^-$ is a closed non-empty set.  Since $S$ has mean curvature $\le -n$ (in the usual smooth sense) and meets $S^-$ locally to the future near any intersection point $x \in S \cap S^-$, we can apply the geometric maximum principle for $C^0$ spacelike hypersurfaces in \cite{AGH98b}. By  \cite[Theorem 3.6]{AGH98b}, for each $x \in S \cap S^-$, there exists an open spacetime neighborhood $U$ of $x$, such that $S \cap U =  S^- \cap U$. (A certain technical condition in Theorem 3.6 is automatically satisfied since $S$ is smooth.)
It follows that $S \cap S^-$ is open and closed in both $S$ and $S^-$, and hence $S = S^-$.  But this means that $S$ admits a future $S$-ray from each of its points. The theorem now follows from Lemma \ref{lem:split}.
\end{proof}



\section{Cosmological time} \label{sec:timefcts}

In this section, we revisit two canonical choices of time function  for spacetimes relevant to cosmology: the cosmological time function and the cosmological volume function. In particular, we prove part 3 of Theorem~\ref{thm:intro1} and Theorem~\ref{thm:intro2} from the Introduction.

\subsection{The cosmological time and volume functions in relation to Hawking's theorem}

Recall that a continuous function $f \colon M \to \bR$ is called a \textit{time function} if it is strictly increasing along future-directed causal curves. The \textit{cosmological time function} is defined as
\begin{align*}
 \tau(p) :=& \sup \left\{ \dg(x,p) \colon x \in M \right\} \\
 =& \sup \left\{ \Lg(\gamma) \colon \gamma \text{ past-directed causal starting at } p \right\}
\end{align*}
Hawking's singularity theorem (point 1 of Theorem~\ref{thm:intro1}) guarantees that $\tau$ is finite under its assumptions. Finiteness alone, however, does not guarantee that $\tau$ is {\it regular} in the sense of \cite{AGH98}, since that additionally requires $\tau \to 0$ along every past-inextendible causal curve. In fact, \cite[Ex.~4.1]{AGH98} shows that regularity can fail even when $\tau$ is finite and the spacetime is globally hyperbolic (we do not know if there exist counterexamples that also satisfy the curvature assumptions of Theorem~\ref{thm:intro1}). Conversely, when $\tau$ is regular, then the spacetime is automatically globally hyperbolic, $\tau$ is a continuous time function (in fact, Lipschitz with two derivatives almost everywhere), and for every point $p$, there is a past-inextendible geodesic whose length realizes $\tau(p)$ \cite[Thm.~1.2]{AGH98}.

Similarly, the formulation of Hawking's theorem using volumes \cite[Thm.~5.2]{GH23} (based on the work of Treude and Grant \cite{TrGr13}) implies that the cosmological volume function
\begin{equation*}
 \tvol(p) := \vol (I^-(p))
\end{equation*}
is finite when the assumptions of Theorem~\ref{thm:intro1} are met, even if $H=n$. As above, we call $\tvol$ regular if $\tvol<\infty$ and $\tvol \to 0$ along every past-inextendible causal curve. Unlike the cosmological time function $\tau$, we prove below that a finite cosmological volume function $\tvol$ on a globally hyperbolic spacetime is automatically regular. On the other hand, while a regular $\tvol$ is a continuous time function, further continuity or differentiability properties like those of $\tau$ have not been established.

\begin{proposition}
 Let $(M,g)$ be a globally hyperbolic spacetime with finite cosmological volume function $\tvol$. Then $\tvol$ is regular. In particular, if $(M,g)$ satisfies the assumptions of Theorem \ref{thm:intro1}, then $\tvol$ is regular.
\end{proposition}

\begin{proof}
 By \cite[Thm.~5.2]{GH23}, the assumptions of Theorem \ref{thm:intro1} guarantee that $\tvol$ is finite-valued. Suppose $\tvol$ is not regular (but still finite valued), then there exists a past-inextendible past-directed causal curve $\gamma : [0,\infty) \to M$ such that $\tvol \circ \gamma(s) \not\to 0$ as $s \to \infty$. This implies by standard measure theory \cite[Thm.~1.2.5]{KrPa08} that
 \begin{equation*}
  \vol \left( \bigcap_{s \in [0,\infty)} I^-(\gamma(s)) \right) = \lim_{s\to\infty} \vol \left(I^-(\gamma(s) \right) = \lim_{s\to \infty} \tvol(\gamma(s)) > 0.
 \end{equation*}
 It follows that there exists some point $p \in \bigcap_{s \in [0,\infty)} I^-(\gamma(s))$. Let $q := \gamma(0)$. Since, by global hyperbolicity, the causal diamond $J(p,q) := J^+(p) \cap J^-(q)$ is compact, in order to avoid a strong causality violation, $\gamma$ must eventually leave  $J(p,q)$ to the past, not to return.  But this contradicts the fact that  $\gamma(s) \in I^+(p)$ for all $s \in [0,\infty)$.
\end{proof}

\subsection{Cauchyness of the level sets of the cosmological time and volume functions} \label{sec:cauchyness}

In \cite[Cor.~2.6]{AGH98}, it is proven that the level sets $S_T$ of the cosmological time function $\tau$ are \textit{future} Cauchy hypersurfaces (i.e.\ any past-inextendible timelike curve starting in $J^+(S_T)$ must intersect $S_T$). In this section, we investigate when the level sets are Cauchy hypersurfaces.

\begin{theorem} \label{thm:taualmostcauchy}
 Let $(M,g)$ be a future timelike geodesically complete spacetime satisfying at least one of the following conditions:
 \begin{enumerate}
  \item $(M,g)$ contains a compact Cauchy hypersurface.
  \item The future causal boundary of $(M,g)$ is spacelike.
 \end{enumerate}
 Then, if the cosmological time function $\tau \colon M \to (0,\infty)$ is regular, its level sets $S_T := \{ \tau = T \}$ are Cauchy hypersurfaces, for all $T>0$.
\end{theorem}

Recall that the future causal boundary is said to be spacelike when there are no proper inclusions between TIPs. We shall only use this assumption in order to apply a Lemma from \cite{WaYi81}. See \cite{GKP72} for more background on causal boundaries.

\begin{proof}
Let $\sigma \colon \bR \to M$ be any inextendible causal curve, and $T >0$. We have to show that $\sigma$ intersects $S_T$. This follows by continuity if there exists $s,t$ such that $\tau(\sigma(s)) \leq T \leq \tau(\sigma(t))$. The existence of $s$ follows from regularity of the cosmological time function. The existence of $t$ can be proven by contradiction as follows.

Suppose that $\tau(\sigma(r)) < T$ for all $r \in \bR$. Choose a Cauchy hypersurface $S$ in the past of $\sigma(0)$ (note that $(M,g)$ is globally hyperbolic by assumption, since its cosmological time function is regular). By standard properties of Cauchy hypersurfaces \cite[Thm.~14.44]{ONeill}, for each $n$ there exists a geodesic $\gamma_n : [0,b_n] \to M$ starting normal to $S$ and ending at $\sigma(n)$, which realizes the Lorentzian distance $\dg(S,\sigma(n))$. For later convenience, we take the $\gamma_n$ to be parametrized by $h$-arclength for some auxiliary complete Riemannian metric $h$ on $M$.

By assumption, the starting points $\gamma_n(0)$ are all contained in some compact set $K$, since:
\begin{enumerate}
 \item If $(M,g)$ has compact Cauchy hypersurfaces, we take $K=S$.
 \item If the causal boundary is spacelike, then by \cite[Lem.~2]{WaYi81}, $K:=\overline{I^-(\sigma)} \cap S$ is compact (here we have used that $I^-(\sigma)$ is a TIP, see \cite[Thm.~2.3]{GKP72}).
\end{enumerate}
By compactness of $K$, the sequence of starting points $\gamma_n(0)$ converges (up to taking a subsequence). Moreover, since $h$ is complete, $\sigma$ is inextendible, and $(M,g)$ is globally hyperbolic (hence non-totally imprisoning), it follows that $b_n \to \infty$. Then, by the limit curve theorem \cite[Thm.~2.51]{Min19}, a subsequence of $(\gamma_n)_{n \in \bN}$ converges uniformly on compact subsets to an inextendible causal limit curve $\gamma_\infty : [0,\infty) \to M$. The construction of $\gamma_\infty$ is depicted in Figure \ref{fig:pf}. Moreover, because the $\gamma_n$ maximize the Lorentzian distance to $S$, so does $\gamma_\infty$ \cite[Thm.~2.61]{Min19}. Then $\gamma_\infty$ is a future $S$-ray and, in particular, it is normal to $S$ and hence timelike.

\begin{figure}
\centering
\begin{tikzpicture}[scale=0.8]
 \draw [thick] plot [smooth, tension=0.9] coordinates { (-4,-1) (-2.5,-0.8) (-1,-1) (0.5,-1.2) (2,-1)};
 \draw [thick] plot [smooth, tension=0.9] coordinates { (-2,1) (-0.5,1.2) (1,1) (2.5,0.8) (4,1)};
 \draw [thick] plot [] coordinates { (-4,-1) (-3,0) (-2,1)};
 \draw [thick] plot [] coordinates { (2,-1) (3,0) (4,1)};
 \draw [thick] plot [smooth, tension=0.9] coordinates {(2.5,1.4) (2,2) (0,4) (-1,5.5) (-1.4,6.2)}; 
 \draw [thick] plot [smooth, tension=0.9] coordinates {(0,4) (-0.2,3) (0,2) (0.2,1) (0,-0.2)}; 
 \draw [thick] plot [smooth, tension=0.9] coordinates {(-1,5.5) (-1.3,3.5) (-1,2) (-0.9,1) (-1,-0)}; 
 \draw [thick] plot [smooth, tension=0.9] coordinates {(-1.85,6.2) (-2.1,4.5) (-1.9,3) (-1.6,1.5) (-1.7,0)}; 
 \draw [thick,dotted] plot coordinates { (-1.4,6.2) (-1.54,6.5)}; 
 \draw [thick,dotted] plot coordinates {(-1.85,6.2) (-1.77,6.55)}; 
 \filldraw[black] (2.5,1.4) circle (2pt) node[anchor=west] {$\sigma(0)$};
 \filldraw[black] (0,4) circle (2pt) node[anchor=south west] {$\sigma(1) = \gamma_1(b_1)$};
 \filldraw[black] (-1,5.5) circle (2pt) node[anchor=south west] {$\sigma(n) = \gamma_n(b_n)$};
 \filldraw[black] (0,-0.2) circle (2pt) node[anchor=north west] {$\gamma_1(0)$};
 \filldraw[black] (-1,-0) circle (2pt) node[anchor=north] {$\gamma_n(0)$};
 \filldraw[black] (-1.7,0) circle (2pt) node[anchor=north east] {$\gamma_\infty(0)$};
 \node at (2,2.5) {$\sigma$};
 \node at (0.3,2.5) {$\gamma_1$};
 \node at (-0.9,3) {$\gamma_n$};
 \node at (-2.5,3.6) {$\gamma_\infty$};
  \node at (3,-0.75) {$S$};
\end{tikzpicture}
\caption{An illustration of the construction of $\gamma_\infty$ in the proof of Theorem \ref{thm:taualmostcauchy}.} \label{fig:pf}
\end{figure}
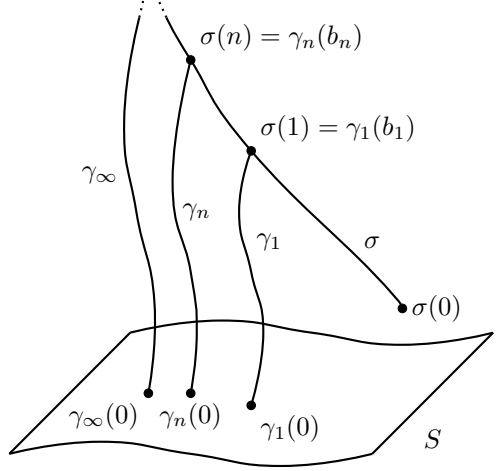

The remainder of the proof consists in showing that the assumption that $\tau$ remains bounded along $\sigma$ implies that $\gamma_\infty$ must have finite length. Notice that by definition of $\tau$, we have $\tau(\sigma(n)) \geq \Lg(\gamma_n)$. Thus $\dg(S,\sigma(n)) = \Lg(\gamma_n) < T$ for all $n$. Then, for any $a \in \bR$,
\begin{equation*}
 \Lg\left(\gamma_\infty \vert_{[0,a]}\right) = \lim_{n \to \infty} \Lg\left(\gamma_n \vert_{[0,a]}\right) \leq T,
\end{equation*}
where $\Lg(\gamma_\infty \vert_{[0,a]}) \geq \lim_{n \to \infty} \Lg(\gamma_n \vert_{[0,a]})$ is always true by upper semicontinuity of the Lorentzian length functional, while the opposite inequality is a consequence of $\gamma_n$ being maximizing for each $n$. We conclude that $\Lg(\gamma_\infty) \leq T$, in contradiction to future timelike geodesic completeness.
\end{proof}

\begin{example}[Not future timelike geodesically complete] \label{ex:cyl}
 Consider the subset of $(1+1)$-dimensional Minkowski spacetime given by $M := \left\{ (t,x) \colon 0 < t < f(x) \right\}$, where
 \begin{equation*}
  f(x) := \begin{cases}
           2+\frac{1}{2}(x-2n) &\text{for } x \in [2n,2n+1), \\
           3 - \frac{1}{2}(x - 2n) &\text{for } x \in [2n+1,2n+2).
          \end{cases}
 \end{equation*}
 Note that the future causal boundary of $(M,g)$ is spacelike (see Figure~\ref{fig:cyl}). The cosmological time function on $(M,g)$ is given by $\tau = t$ and is regular. While the level sets for $t < 2$ are Cauchy hypersurfaces, those for $t \geq 2$ are not. We can also obtain examples with compact Cauchy surfaces by identifying $(t,x) \sim (t,x+2m)$. The resulting spacetimes can be embedded into $(1+2)$-dimensional Minkowski spacetime as a finite cylinder with one ``irregularly cut'' end.
\end{example}

\begin{figure}
\centering
\begin{tikzpicture}
 \draw[thick,dashed]  foreach \k in {0, ..., 3} {(2*\k,2) -- (2*\k+1,2.5) -- (2*\k+2,2)};
 \draw[thick,red]  foreach \k in {0, ..., 3} {(2*\k+0.5,2.25) --  (2*\k+1.5,2.25)};
 \draw[thick,blue] (0,1) -- (8,1);
 \draw[thick,dashed] (0,0) -- (8,0);
 \draw[thick,dotted] (0,0) -- (0,2);
 \draw[thick,dotted] (8,0) -- (8,2);
 \node[anchor=north] at (1,1) {$t=1$};
 \node[anchor=north] at (1,2.25) {$t=\frac{9}{4}$};
 \draw (7,0.4) -- (7.2,0.6) -- (6.8,0.6) -- (7.2,0.2) -- (6.8,0.2) -- cycle;
 \draw[->] (7,0.4) -- (7,0.8);
\end{tikzpicture}
\caption{The spacetime in Example \ref{ex:cyl}, with a Cauchy and a non-Cauchy level set. The dashed lines represent the end of the spacetime. The dotted vertical lines can be either identified, or infinitely many copies of the same picture can be attached.} \label{fig:cyl}
\end{figure}
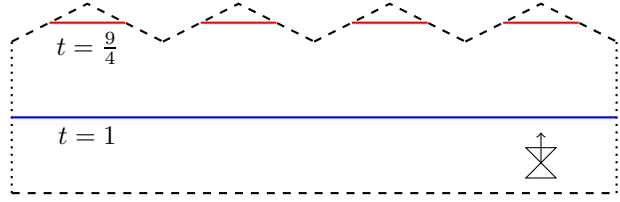

It remains open whether the assumption 1 or 2 in Theorem \ref{thm:taualmostcauchy} is necessary.
The (most obvious) analogue to Theorem \ref{thm:taualmostcauchy} for the cosmological volume function is false.

\begin{proposition}
 There exist spacetimes $(M,g)$ that are future volume complete, i.e.\ $\vol(I^+(p)) = \infty$ for all $p \in M$, and where $\tvol$ is regular, but where not all (non-empty) level sets of $\tvol$ are Cauchy hypersurfaces.
\end{proposition}

\begin{proof}[Proof sketch]
 Such an example can be obtained by modifying Example \ref{ex:cyl} with a conformal factor $\Omega$ such that $\Omega \equiv 1$ for $t < 2$, and $\Omega(t,x) \to \infty$ fast enough as $(t,x) \to (f(x_0),x_0)$, for any $x_0 \neq 2n$. This asymptotic behavior guarantees that $\vol(I^+(p)) = \infty$ for all $p \in M$, since $I^+(p) \cap \{ t>2 \} \neq \emptyset$ for all $p \in M$. On the other hand, the causal curve $\gamma \colon (0,2) \to M,\ s \mapsto (s,0)$ has the property that $I^-(\gamma) \subset \{t<2\}$, from which one can easily see that
 \begin{equation*}
  \lim_{s \to 2} \tvol(\gamma(s)) = \vol (I^-(\gamma)) < \infty,
 \end{equation*}
 implying that $\gamma$ does not reach the level sets for large values of $\tvol$. Other curves do reach those level sets, again thanks to the asymptotic behaviour of $\Omega$.
\end{proof}

\subsection{Mean curvature of the level sets of the cosmological time function} \label{sec:Hsupport}

Proposition \ref{prop:Hoftau} below establishes an upper bound on the mean curvature of the level sets of the cosmological time function $\tau$. Since $\tau$ need not be differentiable everywhere, however, the level sets are in general only $C^0$ hypersurfaces, so the mean curvature bound is in a weak sense.

Such a notion of weak mean curvature bound was introduced by Eschenburg~\cite[Section 4]{Esch89} in the Riemannian case, and developed further in both the Lorentzian and Riemannian cases by Andersson, Howard and the first author in \cite{AGH98b}. At the end of this section, we prove a generalization of the Hawking singularity theorem to $C^0$ hypersurfaces with weak mean curvature bounds.


\begin{definition} \label{def:Hsupport}
 Let $(M,g)$ be a smooth spacetime and $S$ a $C^0$ hypersurface in $M$. We say that $\Sigma$ is a {\it future support hypersurface} to $S$ near $q \in S$, if $q \in \Sigma$ and $\Sigma$ is a smooth spacelike hypersurface such that, in a neighborhood of $q$, $\Sigma \subset J^+(S)$. Moreover, we say that $S$ has {\it mean curvature $H \leq \beta$ at $q$ in the support sense}, if for every $\epsilon > 0$, there exists a future support hypersurface $\Sigma$ to $S$ near $q$ with mean curvature $H_\Sigma \leq \beta + \epsilon$. Finally, we say that that $S$ has {\it mean curvature $H \leq \beta$ in the support sense} if this holds at every point $q \in S$, and {\it $H < \beta$ in the support sense} if for every point $q$, there is a $\delta > 0$ such that $H \leq \beta - \delta$ at $q$ in the support sense.
\end{definition}
Analogously, we define lower mean curvature bounds in the support sense by using past support hypersurfaces. At this point, it is not obvious that $H \leq \beta$ and $H > \beta$ are mutually exclusive; one can show this using the maximum principle \cite[Thm.~1]{Esch89} in an argument similar to the proofs of Theorems~\ref{thm:Hawsupport} and~\ref{thm:Hawsupport2} below.

\begin{proposition} \label{prop:Hoftau}
  Let $(M,g)$ be an $(n+1)$-dimensional spacetime with $\Ric \geq -n\kappa g$, $\kappa \in \bR$, in timelike directions, equipped with a regular cosmological time function $\tau$. Then each level set $S_T := \{ \tau = T\}$ has mean curvature $H_T \leq \beta_{\kappa,T}$ in the support sense, where

\begin{equation}\label{beta}
\beta_{\kappa,T} = \begin{cases}
n \sqrt{|\kappa|} \coth(\sqrt{|\kappa|} T)\, , & \kappa < 0 \\
 \frac{n}{T} \, , & \kappa = 0  \\
n \sqrt{\kappa} \cot(\sqrt{\kappa} T) \, ,  & \kappa > 0  \,.
 \end{cases}
\end{equation}
\end{proposition}

\proof  By Theorem 1.2(iii) in \cite{AGH98},  for every $q \in S_T$
there exists a future directed unit speed timelike geodesic $\gamma: (0,T] \to M$, with $\gamma(T) = q$, which is maximal on each segment and such that
$\tau(\gamma(t)) = t$.   In fact, it is easily seen that each segment $\gamma |_{[t,T]}$ realizes the distance to $S_T$.  It follows that there are no cut points to $\gamma(t)$ on $\gamma |_{[t,T]}$. This in turn implies that each future Lorentzian sphere $\Sigma_r = \{ x \in\ M: d(\gamma(t),x) = r\}$, $r \in (0,T-t]$, is smooth near $\gamma(r+t)$.  Moreover, using that $\gamma$ maximizes to $S_T$, it follows that $\Sigma_{T-t}$ lies locally to the future of  $S_T$ near
$q  \in S_T$, and hence  is an upper support hypersurface for $S_T$ near $q$.

Let $H(r)$ be the mean curvature of $\Sigma_r$ at $\gamma(r+t)$. The Raychaudhuri equation and Ricci curvature condition imply
$$
\frac{d \mathcal{H}}{dr} \le -\kappa - \mathcal{H}^2 \, ,
$$
where $\mathcal{H} = H/n$. Then by standard comparison results (see e.g. \cite[Section~2]{Esch87}, \cite[Section 1.6]{Karch89}, \cite[Section 3]{TrGr13}), $H(T-t) \le  \beta_{\kappa,T-t}$  for all values of $\kappa$.  The result now follows by letting $t \to 0$. \qed

\begin{remark}
 Proposition \ref{prop:Hoftau} gives an upper bound on the mean curvature of $S_T$, which becomes more restrictive  for large $T$. For small $T$, the mean curvature can diverge, as one might expect from a big bang. By a time reversed argument, one can also prove a lower bound on the mean curvature at those points $q \in S_T$ that are starting points of a future $S_T$-maximizing geodesic $\gamma$. (Such geodesics $\gamma$ are easily constructed since $S_T$ is future Cauchy.) The longer $\gamma$ maximizes, the more restrictive the lower bound. In particular, if, in the case $\kappa < 0$ or $\kappa = 0$, $S_T$ admits a complete future $S_T$-ray, then the mean curvature of the level set $S_t$ at the intersection with the ray satisfies
\begin{equation} \label{eq:Hray}
  - n \sqrt{|\kappa|} \leq H_t \leq \beta_{\kappa,t}
 \end{equation}
 for all $t \geq T$. One can extend such a ray to the past by attaching
 a geodesic that realizes $\tau(q)$, and, from the maximizing property of each of these geodesics, obtain a (past incomplete, future complete) timelike line (i.e. an inextendible globally maximizing timelike geodesic) along which \eqref{eq:Hray} holds for all $t>0$.
\end{remark}

We end this paper by proving two variations of Hawking's singularity theorem for mean curvature in the support sense.

\begin{theorem} \label{thm:Hawsupport}
 Let $(M,g)$ be an $(n+1)$-dimensional globally hyperbolic spacetime with $\Ric \geq n g$ in timelike directions. Suppose that $S$ is a $C^0$ Cauchy hypersurface in $(M,g)$ with mean curvature $H \geq \beta$ in the support sense, for some constant $\beta > n$. Then $(M,g)$ is past timelike geodesically incomplete. Moreover, the length of any past-directed timelike curve starting on $S$ can be no greater than
 \begin{equation*}
  \ell = \coth^{-1} \left( \frac{\beta}{n} \right).
 \end{equation*}

\end{theorem}

\begin{proof}
 It suffices to prove that there do not exist any points $p \in I^-(S)$ with $T := \dg(p,S) > \ell$. Suppose that there does exist such a point $p$. Then there exists a unit-speed timelike geodesic $\gamma : [0,T] \to M$ from $p$ to $S$ of length $\Lg(\gamma) = T$. Since $\gamma \vert_{[t,T]}$ realizes the distance to $S$, then as in the proof of Proposition \ref{prop:Hoftau}, the future Lorentzian sphere $\Sigma_{T-t} = \{ x \in\ M: d(\gamma(t),x) = T-t\}$ is a smooth future support hypersurface to $S$ at $q := \gamma(T)$, with mean curvature $H(T-t) \leq \beta_{-1,T-t} = n \coth(T-t)$. In particular, for any $\epsilon >0$, there is a small enough  $t$ such that
 \begin{equation*}
  H(T-t) \leq n \coth(T-t) < n \coth (\ell) - \epsilon = \beta - \epsilon.
 \end{equation*}
 We proceed to derive a contradiction from this fact. By assumption, $H \geq \beta$ in the support sense, so for $\epsilon$ as above, there also exists a past support hypersurface $\Sigma'$ to $S$ at $q$ with mean curvature $H' \geq \beta - \epsilon$. Then
 \begin{equation} \label{eq:Hineq}
  H(T-t) < \beta - \epsilon  \leq H',
 \end{equation}
 and by the maximum principle \cite[Thm.~1]{Esch89}, near $q$ it holds that $\Sigma_{T-t} = \Sigma'$ and $H(T-t) = H'$. This is clearly impossible, since the first inequality in \eqref{eq:Hineq} is strict. We conclude that there cannot exist any points $p$ with $\dg(p,S) > \ell$, as desired.
\end{proof}

Let us comment here on which other parts of the paper can be generalized to mean curvature bounds in the support sense. Theorem \ref{thm:newHawking2} can be generalized by carefully applying the geometric maximum principle \cite{AGH98b} in the case where neither of the two surfaces is smooth, cf.~\cite[Prop. 4.3]{GaVe15}. Remark \ref{rem:rescale} about behaviour of the curvature under rescalings also applies to Theorem \ref{thm:Hawsupport}. Moreover, one can prove a version of Theorem \ref{thm:Hawsupport} for $\Ric \geq 0$ in timelike directions and $H \geq \beta >0$ in the support sense, in which case the length bound becomes $\ell = \frac{n}{\beta}$. A version for $\Ric \geq -n g$ is also possible, but in that case, the spacetime has bounded timelike diameter, so one has timelike incompleteness even without any mean curvature assumptions \cite[Thm.~11.9]{BEE}. What is more interesting is that, for $C^0$ Cauchy hypersurfaces, assuming $H > n$ in the support sense does not imply $H \geq \beta > n$, even when $S$ is compact, since the $\delta$ in Definition \ref{def:Hsupport} need not be continuous. It thus remains open whether part 1 of Theorem~\ref{thm:newHawking} holds in the support sense. At least, when assuming $S$ compact with $H > n$, we still get the following, weaker statement.

\begin{theorem} \label{thm:Hawsupport2}
 Let $(M,g)$ be an $(n+1)$-dimensional globally hyperbolic spacetime with $\Ric \geq n g$ in timelike directions. Suppose that $S$ is a compact $C^0$ Cauchy hypersurface in $(M,g)$ with mean curvature $H > n$ in the support sense. Then $(M,g)$ is past timelike geodesically incomplete.
\end{theorem}

\begin{proof}
 It is well-known that a compact Cauchy hypersurface $S$ always admits a past $S$-ray, that is, a past-inextendible geodesic $\gamma$ from $S$ that realizes the distance to $S$ on its entire domain. Indeed, a constructive proof can be given by a limiting procedure, as in the proof of Theorem \ref{thm:taualmostcauchy}. We show that, under our curvature assumptions, every past $S$-ray is incomplete.

 Suppose that $\gamma : [0,\infty) \to M$ is a complete past $S$-ray, parametrized by Lorentzian arc-length. Then, by a similar argument as in the proof of Proposition~\ref{prop:Hoftau}, the Lorentzian sphere $\Sigma_{t} := \{ x \in\ M: d(\gamma(t),x) = t \}$ is a smooth future support hypersurface to $S$ at $q := \gamma(0)$, with mean curvature $H(t) \leq \beta_{-1,t} = n \coth (t)$. But at the same time, $H > n$ in the support sense, meaning that there exists $\delta > 0$ such that $H \geq n+\delta$ at $q$ in the support sense. Thus there is a past support hypersurface $\Sigma'$ at $q$ with mean curvature $H' \geq n + \delta - \epsilon$. Choosing $\epsilon = \delta/2$ and $t$ large enough, we obtain $H(t) < H'$. Therefore, we can apply the maximum principle to $\Sigma_t$ and $\Sigma'$ (as we did when proving Theorem~\ref{thm:Hawsupport}) to arrive at a contradiction. We conclude that every past $S$-ray in $(M,g)$ is incomplete, and since there exists at least one, the spacetime is past timelike geodesically incomplete.
\end{proof}

\bigskip

\noindent \textbf{Acknowledgements.} LGH thanks Miguel S\'{a}nchez for interesting discussions at the initial stages of this work. The authors are grateful to an anonymous reviewer for helpful feedback. The work of GJG was partially supported by the Simons Foundation, under Award No. 850541.

\bibliographystyle{abbrv}
\bibliography{mybib}

\end{document}